\def\year{2022}\relax
\documentclass[letterpaper]{article}
\usepackage{aaai22} 
\usepackage{times} 
\usepackage{helvet} 
\usepackage{courier} 
\usepackage[hyphens]{url} 
\usepackage{graphicx} 
\usepackage{graphicx} 
\usepackage{natbib} 
\usepackage{caption} 
\DeclareCaptionStyle{ruled}%
{labelfont=normalfont,labelsep=colon,strut=off}
\usepackage{xcolor}
\usepackage{thmtools}
\usepackage{thm-restate}
\usepackage{complexity}

\title{Generalized Stochastic Matching}

\author{Alireza Farhadi, Jacob Gilbert, MohammadTaghi Hajiaghayi}

\affiliations{ University of Maryland \\ farhadi@cs.umd.edu, jgilber8@umd.edu, hajiagha@cs.umd.edu}

\usepackage{amssymb, amsthm, amsmath}
\DeclareMathOperator*{\argmax}{arg\,max}
\usepackage{algpseudocode, algorithm}
\usepackage{mathtools}

\newtheorem{theorem}{Theorem}
\newtheorem{lemma}[theorem]{Lemma}
\newtheorem{definition}[theorem]{Definition}
\newtheorem{observation}[theorem]{Observation}

\newcommand{\SM}[0]{\emph{stochastic matching}}

\begin{document}

\maketitle

\begin{abstract}
In this paper, we generalize the recently studied \SM\ problem to more accurately model a significant medical process, kidney exchange, and several other applications.  Up until now the \SM\ problem that has been studied was as follows: given a graph $G = (V, E)$, each edge is included in the $\emph{realized}$ sub-graph $\mathcal{G}$ of $G$ mutually independently with probability $p_e$, and the goal is to find a degree-bounded sub-graph $Q$ of $G$ that has an expected maximum matching that approximates the expected maximum matching of $\mathcal{G}$. This model does not account for possibilities of vertex dropouts, which can be found in several applications, e.g. in kidney exchange when donors or patients opt out of the exchange process as well as in online freelancing and online dating when online profiles are found to be faked. Thus, we will study a more generalized model of \SM\ in which vertices and edges are both $\emph{realized}$ independently with some probabilities $p_v, p_e$, respectively, which more accurately fits important applications than the previously studied model.

We will discuss the first algorithms and analysis for this generalization of the \SM\ model and prove that they achieve good approximation ratios.
In particular, we show that the approximation factor of a natural algorithm for this problem is at least $0.6568$ in unweighted graphs, and $1/2 + \epsilon$ in weighted graphs for some constant $\epsilon >0$. We further improve our result for unweighted graphs to $2/3$ using edge degree constrained subgraphs (EDCS).
\end{abstract}

\section{Introduction}
The \SM\ problem has been used to model kidney exchange in several research papers in recent years, and in this paper, we generalize this model to better suit the needs of kidney exchange and other applications. Kidney exchange is an important medical procedure that is utilized to increase the amount of possible successful kidney transplants between patients and donors for hundreds of donor-patient pairs in the U.S. each year. This medical process occurs when an incompatible kidney donor-patient pair matches with another incompatible pair such that the donors are swapped to become compatible pairs with the patients.  Unfortunately, compatibility medical testing can require patients and donors to be hospitalized and are expensive.  Thus, minimizing the amount of compatibility tests while maximizing compatible exchanges is an important problem in the medical world.
Moreover, in this paper we also account for the possibility that a patient or donor may decide to drop out of the exchange due to health conditions or at their discretion at any point throughout the  months' long process.  Therefore, while prior papers only considered donor-patient compatibility, we will also consider potential dropouts from the exchange process on top of compatibility.

In our proposed \SM\ model of kidney exchange, each donor-patient pair is represented by a vertex in the graph. Edges in this graph $G$ represent donor-patient pairs that may be compatible for exchange. 
Only a subset of the edges are found to be compatible through medical records and testing, and this subset forms a realized sub-graph $\mathcal{G}$ of possible successful exchanges.  We say these edges are $\emph{realized}$, i.e. appear in $\mathcal{G}$, with some probability $p_e$.  Similarly, a vertex is $\emph{realized}$ with probability $p_v$ if the pair does not dropout during the exchange process.  In our generalized model, an edge can only be included in $\mathcal{G}$ if both of its vertices are realized as well. A maximum matching algoritm seeks to pair vertices connected by an edge of the graph together to create the maximum amount of matches. So, a maximum matching of $\mathcal{G}$ represents maximized compatible kidney exchanges.  As mentioned, medical tests for compatibility are expensive, and so $\emph{querying}$ the edges of $G$ to see if they were realized should be kept to a minimum.  Without knowing the sub-graph $\mathcal{G}$, the goal of the \SM\ problem is to find some degree-bounded sub-graph $Q$ with an expected
maximum matching of realized edges that has a size approximately that of the actual maximum matching of $\mathcal{G}$.  We will state and prove the existence of the first bounds for the approximation ratio achieved with this model of kidney exchange in which vertices and edges may be dropped from the original graph.

\subsection{Additional Applications}
With our generalized \SM, in addition to kidney exchange we can model the freelancing industry comprised of freelance workers and their potential employers.  In modern freelancing, workers may have online profiles on websites that businesses can look through to find freelancers with compatibility for a job or project.  Unfortunately, a large amount of fake profiles and fake job offerings plague these websites.  Finding out profiles and jobs are fake costs time and money from those who hired the fake profiles or the freelancers who took up a fake job.  Therefore, in the online freelancing problem the goal is to maximize matchings between jobs and freelancers while minimizing the amount of queries to freelancers and employers needed to match real freelancer profiles to real job offers.

To model online freelancing with \SM, profiles and companies will make up the vertices of some graph $G$, and there is an edge between profiles and companies if a freelancer fits the qualifications for a company's job opening.  Edges may be weighted by the amount a company will pay a freelancer for the job, or remain unweighted if all jobs are nearly equally valuable. Each vertex is realized with some probability $p_v$  as long as the online profile or company is real.  Each edge $(u, v)$ is realized as long as both of the vertices $u$ and $v$ are realized.  As in kidney exchange, $\emph{queries}$ of edges/vertices are expensive since they require profile reviews and lengthy communications, but in this version of the problem  only vertices have a realization probability while edges are always realized if both of its vertices have been realized. Online dating is a very similar scenario with potential fake vertices, but edges between dating profiles may also drop out if a match does not lead to a relationship.  In our results later in the paper, we will state and prove bounds to the approximation ratio achieved for weighted graphs for the freelancing and dating model.  


Besides the aforementioned applications, the problem is significant from a computer science theory perspective as a discussion of graph sparsification.  We will show that a simple, well-studied algorithm provides a sparse sub-graph with a good approximation of the expected maximum matching of the original graph for our generalized version of stochastic matching. This sub-graph will conform to a tight restriction: any vertex has at most constant $O(1)$ degree.

\subsection{Generalized Stochastic Matching Model}
As discussed, the kidney exchange problem may be modeled with our proposed generalized \SM\ model. In the stochastic setting, we have a random sub-graph, the $\it{realized}$ sub-graph, of some given graph, and we want to approximate some property of the realized sub-graph.
\begin{definition}
    Given fixed parameters $p_v, p_e \in (0, 1]$ and weighted or unweighted graph $G = (V, E)$ with vertex set $V$ and edge set $E \subset V^2$, let graph $\mathcal{G} = (\mathcal{V}, \mathcal{E})$ be a sub-graph of $G$ such that any vertex $v \in V$ is in $\mathcal{V}$ mutually independently randomly with probability $p_v$ and any edge $e = (u, v) \in E$ is in $\mathcal{E}$ mutually independently randomly with probability $p_e$ if $u, v \in \mathcal{V}$. We call $\mathcal{G}$ the \emph{realized sub-graph} of $G$.
\end{definition}

\begin{definition}
Given weighted graph $G = (V, E, W)$ where $W$ is a set of edge weights, let $w_e \in W$ be the weight of edge $e \in E$.  Define $M(G)$ to be the maximum weighted matching of $G$; furthermore, let $\mu(G) := \sum_{e \in M(G)} w_e$ be the weight of the the maximum matching of $G$.
\end{definition}

In the \SM\ problem, we want to find a sparse sub-graph of a given graph such that the realized portion of the sparse sub-graph approximates the maximum weighted matching of the realized sub-graph.
More formally, given a graph $G$ with $n$ vertices, we want to find a sub-graph $Q = (V, E_Q)$ that satisfies the following two conditions:

\begin{enumerate}
    \item Let $\mathcal{Q} = Q \cap \mathcal{G}$, then the approximation ratio $\mathbb{E}[\mu(\mathcal{Q})]/\mathbb{E}[\mu(\mathcal{G})]$ is as large as possible.
    \item The degree of $Q$ is $O(1)$.  Specifically, the maximum degree of any vertex in $Q$ may be bounded by a constant determined by $p_v, p_e$ but not $n$.
\end{enumerate}

So, if we can find such a sub-graph $Q$, then we may query the $O(n)$ edges of $Q$ instead of doing expensive queries to all $O(n^2)$ edges of $G$ to find out which were realized. However, finding such a sparse sub-graph and proving it has a large approximation ratio is non-trivial.

\subsection{Related Work}
The less generalized version of \SM\ in which all vertices are realized with probability $1$ was first introduced by \cite{DBLP:conf/sigecom/BlumDHPSS15} primarily to model the kidney exchange setting. In this paper, the authors showed positive empirical results on simulated and real data from the United Network for Organ Sharing in which stochastic matching algorithms resulted in a good approximation of the optimal solution. This problem has been extensively studied since then \cite{DBLP:conf/sigecom/AssadiKL16,DBLP:conf/soda/YamaguchiM18, DBLP:conf/sigecom/BehnezhadR18, DBLP:conf/sagt/BehnezhadD0HR19}. The first discussion of this less generalized problem by \cite{DBLP:conf/sigecom/BlumDHPSS15} achieved an approximation ratio of $(1/2 - \epsilon)$ in unweighted graphs, and then \cite{DBLP:conf/sigecom/AssadiKL17} broke the half approximation barrier with an approximation ratio of .5001. This bound was later improved by \cite{DBLP:conf/soda/BehnezhadFHR19} to .6568 and by \cite{DBLP:conf/soda/AssadiB19} to $(2/3 - \epsilon)$. Afterwards, \cite{DBLP:conf/stoc/BehnezhadDH20} and \cite{DBLP:journals/corr/abs-2004-08703} both built on the analysis of the algorithm proposed by \cite{DBLP:conf/soda/BehnezhadFHR19} to further improve approximation ratios for unweighted and weighted graphs to $(1- \epsilon)$, respectively. We adapt this same algorithm as Algorithm \ref{alg:nonadaptive1} below to fit our model.
\begin{algorithm}
  \caption{An algorithm for the generalized stochastic matching problem.}
  \label{alg:nonadaptive1}
  \begin{algorithmic}[1]
  	\Statex \textbf{Input:} Input weighted graph $G=(V, E)$ and realization probabilities $p_v, p_e \in [0, 1]$.
  	\Statex \textbf{Parameter:} $R := \frac{2000 \log(1/\epsilon) \log(1 / (\epsilon p_v^2p_e))}{\epsilon^4 p_v^2 p_e}$
  	\State $Q \gets (V, \emptyset)$
	\For{$r = 1, \ldots, R$}
		\State Construct a sample $\mathcal{G}_r=(\mathcal{V}_r, \mathcal{E}_r)$ of $G$, where any vertex $v \in V$ appears in $\mathcal{V}_r$ independently with probability $p_v$, and each edge $e \in E$ connecting vertices $u, v$ appears in $\mathcal{E}_r$ independently with probability $p_e$ if and only if $u, v \in \mathcal{V}_r$.
		\State Add the edges in maximum weighted matching $M(E_r)$ of $\mathcal{G}_r$ to $Q$.
	\EndFor
	\State Query the edges in $Q$ and report the maximum weighted matching of it.
  \end{algorithmic}
\end{algorithm}

\subsection{Our Results}
  In the Crucial Edges and Unweighted Approximation section, we will achieve and prove a $.65$ approximation ratio for unweighted stochastic matching, i.e. the kidney exchange model. By adapting the analysis techniques of \cite{DBLP:conf/soda/BehnezhadFHR19} for our new generalization of stochastic matching, we will prove the following theorem and lower bound for the unweighted case:
 
\begin{theorem}
\label{thm:unweighted}
    For unweighted graph $G$, constant $\epsilon > 0$, vertex and edge realization probabilities  $p_v, p_e \in (0, 1]$, there is an algorithm to find an $O_{\epsilon, p}(1)$\footnote{We use $O_{\epsilon, p}(.)$ to hide the dependency on $\poly(\epsilon, p_v, p_e)$.}-degree subgraph $Q$ of $G$ such that $\mathbb{E}[\mu(\mathcal{Q})]/\mathbb{E}[\mu(\mathcal{G})] \ge .6568 - \epsilon$.
\end{theorem}

In the Weighted Approximation section, we further consider the weighted stochastic matching problem, i.e. the freelancing model. Here, we prove the following bounds:

\begin{theorem}
\label{thm:weighted}
        For weighted graph $G$, constant $\epsilon > 0$, vertex and edge realization probabilities  $p_v, p_e \in (0, 1]$, there is an algorithm to find an $O_{\epsilon, p}(1)$-degree subgraph $Q$ of $G$ such that $\mathbb{E}[\mu(\mathcal{Q})]/\mathbb{E}[\mu(\mathcal{G})] \ge .501- \epsilon$.
\end{theorem}

One important distinction between our generalizations for freelancing and kidney exchange from prior work is that some edges are no longer realized completely independently. Specifically, if a vertex is not realized in our model, then every edge connected to it is also not realized.  This dependence sets our model apart from previous \SM\ papers.  Thus, our techniques will not utilize independent realizations of certain edges, a property that both \cite{DBLP:conf/stoc/BehnezhadDH20} and \cite{DBLP:journals/corr/abs-2004-08703} have relied on before. We also improve our bound for unweighted graphs to $(2/3-\epsilon)$ in the EDCS $2/3$ Approximation section using edge degree constraint sub-graphs (EDCS).

\begin{theorem}
\label{thm:unweighted-edcs}
For unweighted graph $G$, constant $\epsilon > 0$, vertex and edge realization probabilities $p_v, p_e \in (0, 1]$, there is an algorithm to find an $O_{\epsilon, p}(1)$-degree subgraph $Q$ of $G$ such that $\mathbb{E}[\mu(\mathcal{Q})]/\mathbb{E}[\mu(\mathcal{G})] \ge 2/3 - \epsilon$.
\end{theorem}

While the bound of Theorem $\ref{thm:unweighted-edcs}$ currently dominates that of Theorem $\ref{thm:unweighted}$ for unweighted graphs,
future improvements to Algorithm \ref{alg:nonadaptive1} and its analysis will eventually most likely overtake the $2/3$ approximation ratio provided by the EDCS approach.

\section{Algorithm \ref{alg:nonadaptive1} Analysis}
\label{sec:analysis}
Before we can prove our main results, we must introduce the concept of fractional matchings and the related procedures we use to build these fractional matchings. Constructing an integral matching directly on our sparse sub-graph $Q$ is difficult since we want a good approximation ratio in expectation without directly knowing $\mathcal{G}$. Instead, we can relax our matching requirements to allow assigning fractional values to edges of our matching, and then later show that the fractional matching serves as proof of the existence of a integral matching of the same approximation ratio.

\subsection{Fractional Matchings}
In order to prove Theorems \ref{thm:unweighted} and \ref{thm:weighted}, we will find a fractional matching $x$ of $\mathcal{Q}$ that achieves a $.6568 - \epsilon$ approximation ratio and $.501 - \epsilon$ approximation ratio, respectively. In an integral matching, each vertex can only be matched to one other vertex.  Alternatively, one can think of an integral matching as assigning a value of either 1 or 0 to every edge such that no vertex has two incident edges with value 1. A fractional matching $x$ provides more flexibility in analysis than an integral matching since it allows assigning fractional values $x_e \in [0, 1]$ to edge $e$ such that for any vertex, $x_v := \sum_{v \in e} x_e \leq 1$. Once we have our fractional matching and prove that it achieves our target approximation ratios, we will use the following folklore lemma to claim the existence of an integral matching $y$ that achieves the same approximation ratio to complete the proofs of Theorem $\ref{thm:unweighted}$ and Theorem \ref{thm:weighted}.  Note that in the following lemma, Lemma $\ref{lem:frac2int}$, given graph $G = (V, E)$ and subset $U \subseteq V$, we use $E(U)$ to refer to the edges of the induced sub-graph on $G$ by $U$ which includes every edge $(u, v) \in E$ such that $u, v \in U$.

\begin{lemma}
\label{lem:frac2int}
    Let $x$ be a fractional matching, $\epsilon > 0$ be a constant, and $G = (V, E)$ be an edge weighted graph where $w_e$ is the weight of edge $e \in E$. If for all $U \subseteq V$ such that $|U| \leq 1/\epsilon$ it is true that $\sum_{e \in E(U)} x_e \leq \lfloor |U| / 2 \rfloor$, then $G$ has an integral matching $y$ such that $\sum_{e \in E} w_e \cdot y_e  \geq (1-\epsilon) \sum_{e \in E} w_e \cdot x_e $. 
\end{lemma}

Proof of Lemma \ref{lem:frac2int} and further discussion about fractional matchings can be found in \cite{DBLP:conf/soda/BehnezhadFHR19} in section 2.2. Now we see that to prove Algorithm \ref{alg:nonadaptive1} provides a good expectecd maximum matching, a fractional matching $x$ on $Q$ will need to satisfy the requirements of Lemma \ref{lem:frac2int} in addition to achieving the target approximation ratio. To create such a matching, we will combine two smaller matchings over two disjoint sets of edges, a set of non-crucial edges and a set of crucial edges.  Each edge will be classified as non-crucial or crucial based on the probability that the edge appears in the maximum matching of $\mathcal{G}$.  

\begin{definition}
    For edge $e \in E$, we define $q_e := \Pr(e \in M(\mathcal{G}))$ as the probability that $e$ appears in the maximum weighted matching of realized sub-graph $\mathcal{G}$, and we will refer to $q_e$ as the $\emph{matching probability}$ of edge $e$\footnote{Given a realization, we can assume that edges belong to the maximum weighted matching are unique. These edges can be the edges returned by an arbitrary deterministic algorithm.}.  Additionally, for vertex $v \in V$, let $q_v := \sum_{e \ni v} q_e$.  For a vertex $v$ and subset $X \subseteq E$, $q(X) := \sum_{e \in X} q_e$ and $q^X_v := \sum_{e: e\in X, v \in e}q_e$.
\end{definition}

\begin{definition}
\label{def:threshold}
    Let threshold $\tau = \frac{\epsilon^3 p_v^2 p_e}{20 \log(1/\epsilon)}$, then edge $e$ is $\emph{crucial}$ if $q_e \geq \tau$ and $\emph{non-crucial}$ if $q_e < \tau$.  We will use $C$ to denote the set of crucial edges and $N$ to denote the set of non-crucial edges.
\end{definition}
The matchings over non-crucial and crucial edges will be constructed with procedures analyzed below.  When creating these procedures, we will have to keep a few things in mind about our new model.  First, edges are only realized if their incident vertices are realized, and so they can be thought of as having a realization probability of not just $p_e$ but $p_v^2p_e$.  With this in mind, when we sort edges into non-crucial and crucial sets, we make sure our threshold incorporates this realization probability $p_v^2 p_e$ in Definition \ref{def:threshold}.  Furthermore, note that in Algorithm $\ref{alg:nonadaptive1}$, the product $p_v^2p_e$ makes an appearance in the number of iterations.  One of the main reasons is to easily relate the number of rounds of the algorithm to our matching probability threshold $\tau$.

Another quirk of the new model is that there is a correlation in realization probabilities of adjacent edges that share an incident vertex. In step 2 of the upcoming non-crucial edge procedure, we scale down our fractional matching by a factor of $p_v$ to account for this.
The first procedure we discuss will create a near-optimal matching on the non-crucial edges using the following useful observation.
\begin{observation}
\label{obs:expmatch}
$\mathbb{E}[\mu(\mathcal{G})] = \sum_{e \in E} w_e q_e$.
\end{observation}
Essentially, we will use matching probabilities as the assigned values in our fractional matching because the expected size of the maximum matching of $\mathcal{G}$ is just the sum of edge weights times matching probabilities. Since we don't actually know the matching probabilities exactly, we will assign the value of $f_e$, defined as the fraction of times edge $e$ appears in a maximum matching of an iteration of Algorithm $\ref{alg:nonadaptive1}$ out of $R$, the total number of iterations. 

\begin{definition}
   For an edge $e \in E$, let $k$ be the number of times $e$ appeared in a maximum weighted matching of a sampled graph $\mathcal{G}_i$ for $1 \leq i \leq R$ during Algorithm $\ref{alg:nonadaptive1}$. We define $f_e := \frac{k}{R}$. 
\end{definition}

Observe that from the above definition, $\mathbb{E}[f_e] = q_e$, which is what we wanted.


\subsection{Non-crucial Edge Procedure}
Given $Q = (V, E_Q)$ from Algorithm \ref{alg:nonadaptive1} with realized sub-graph $\mathcal{Q} = (\mathcal{V}, \mathcal{E}_Q)$. Let $\tilde x_e = 0$ for all edges $e \in E$. Then,
\begin{enumerate}
\label{proc:non-crucial}
	\item For any realized edge $e \in \mathcal{E}_Q \cap N $, set $\tilde x_e$ to be 
	\[\min\{f_e/(p_v^2p_e), 2\tau/(p_v^2p_e) \}.\]
	\item Let $s_e$ be the {\em scaling-factor} of $e$ where the default value is $s_e = 1$. For each vertex $v \in V$ and edge $e$ incident to $v$, set $s_e$ $$s_e = \min\Big\{s_e, \max\{q^N_v, \epsilon\}/( p_v \cdot \sum_{e \ni v} \tilde{x_e})\Big\}.$$
	Note that this step may be done for each vertex in an arbitrary order.
	\item Finally scale down the fractional matching with $s_e$.  So for all edges $e$, let $x_e := \tilde x_e \cdot s_e$.
\end{enumerate}

  By definition of a fractional matching, it is required that $x_v \leq 1$ for any vertex $v \in V$.  So, each vertex can be thought to have a ``budget'' of size 1.  We want to make sure the non-crucial edge procedure leaves some remaining budget for the crucial edges later. To this end, we set the scaling-factors in step 2 to have a factor of $\sum_{e \ni v} \tilde{x_e}$ in the denominator. Moreover, we place a $p_v$ in the denominator because vertices can only be matched if they are first realized. Altogether since $q_v \leq p_v$, in expectation the scaling-factor $q_v^N/( p_v \cdot \sum_{e \ni v} \tilde{x_e})$ should keep $x_e$ below 1. Note that in the actual definition of the scaling-factors we have a max over $q_v^N$ and $\epsilon$, and so, we must make sure it is small enough to stay within the vertex budget even with $\epsilon$. So as discussed, the properties of the following lemma prove an upper bound on the size of the fractional matching per vertex, i.e. some budget remains for crucial edges. Additionally, the first property of the following lemma, Lemma \ref{lem:nsupport}, proves that the non-crucial edge matching satisfies the requirements of Lemma \ref{lem:frac2int}, which as mentioned will be applied in our last step of this analysis to go from our fractional matching to an integral matching. The proof of Lemma \ref{lem:nsupport} and other missing proofs of this section are available in the Appendix.

\begin{lemma}
\label{lem:nsupport}
    Given graph $G = (V, E)$, constant $\epsilon \in (0, 1]$, and fractional matching $x$ from the non-crucial edge procedure:
    \begin{enumerate}
        \item $\forall U \subseteq V$ such that $|U| \leq 1/\epsilon$, $\sum_{e \in E(U)} x_e \leq \epsilon \lfloor |U|/2 \rfloor$.
        \item $\forall v \in V$, $x_v \leq \max\{q_v, \epsilon\}/p_v$.
    \end{enumerate}
\end{lemma}

Now, we have proven that the non-crucial edge procedure is not only a fractional matching, but leaves some room to grow when we discuss crucial edges.  However, as previously discussed, we must now also prove that the current matching is nearly optimal in size. The following definition will be useful to see this.

\begin{definition}
 For $X \subset E$, define $\varphi(X) := \sum_{e \in X} w_e \cdot q_e$ as the \emph{expected matching weight} of $X$.  Also, for a vertex $v$ and subset $X \subset E$, we define $\varphi_v^X := \sum_{e \in X, v \in e} w_e \cdot q_e$.
\end{definition}

By Observation \ref{obs:expmatch}, the expected size of the maximum mathcing on $\mathcal{G}$ is $\varphi(E)$. So in Lemma \ref{lem:ns1}, we will show that matching $x$ over non-crucial edges of our sparse sub-graph is within a $(1-\epsilon)$ factor of $\varphi(N)$. To begin, we will bound the size of $\tilde{x_e}$ from the first step of the non-crucial edge procedure.

\begin{lemma}
\label{lem:ns1}
    $\mathbb{E}\left[\sum\limits_{e \in \mathcal{E}_Q \cap N}w_e \cdot \tilde{x_e} \right] \geq (1-\epsilon)\varphi(N)$.
\end{lemma}

We then show that the scaling factor is at least $(1-5\epsilon)$ with high probability. Therefore, we will know that the scaling factor does not decrease the size of the fractional matching too greatly.

\begin{lemma}
\label{lem:ns2}
    For vertex $v$ of realized edge $e \in \mathcal{E}_Q \cap N$ with probability at least $1-2\epsilon$, it is true that $\max\{q_v^N, \epsilon\}/(p_v \cdot \sum_{e \ni v} \tilde{x_e}) \geq 1 -5\epsilon$.
\end{lemma}

\begin{proof} [Proof of Lemma \ref{lem:ns2}]
    Note that the following proof is an updated proof from Claim A.5, \cite{DBLP:conf/soda/BehnezhadFHR19} for the new model. In order to prove this lemma, we will need to relate the numerator of the scaling factor, $q_v^N$, to the denominator of the scaling factor, $\sum_{e \ni v} \tilde{x}_e$.  Let $\tilde{x_v} := \sum_{e \ni v} \tilde{x}_e$.  To create this relation, we will analyze $f_e$ since it is related to $q_e$ and $\tilde{x_e}$.  Note from the lemma's statement, we have already specified some incident edge $e = (u, v) \in \mathcal{E}_Q \cap N$. If $e$ is not realized, we know that $x_e$ will just be 0 and does not contribute to the matching. Therefore, moving forward we assume that $e$ is realized.  We will use $e_1, e_2, ..., e_k \in E_Q \cap N$ to denote the rest of the edges incident to $v$, realized or unrealized. Now, let $f_v^N := \sum_{e_i} f_e$; we will begin the proof by showing that $f_v^N$ is a good approximation of $q_v^N$. 
    
    Trivially, since $\mathbb{E}[f_e] = q_e$, $\mathbb{E}[f_v^N] \leq q_v^N$.  Intuitively, $f_v^N$ is the fraction representing the number of matchings in which $v$ is matched by some edge $e_i$ in a round of Algorithm $\ref{alg:nonadaptive1}$ divided by $R-1$ since we do not count the round $v$ is matched by $e$. So, by Hoeffding's inequality we have
    \begin{align*}
        \Pr(f_v^N - q_v^N \geq \epsilon^2) &\leq \Pr(f_v^N - \mathbb{E}[f_v^N] \geq \epsilon^2) \\ &\leq exp(-2(R-1)\epsilon^4) \leq \epsilon.
    \end{align*}
    So far we have that with probability at least $1-\epsilon$, $f_v^N - q_v^N \leq \epsilon^2$, and this inequality implies that with probability at least $1-\epsilon$, $\max\{f_v^N, \epsilon\} - \max\{q_v^N, \epsilon\} \leq \epsilon^2$. Therefore, with probability at least $1-\epsilon$,
    \begin{align}
    \max\{f_v^N, \epsilon\} \leq \max\{q_v^N, \epsilon\} + \epsilon^2  &\leq (1+\epsilon)\max\{q_v^N, \epsilon\} . \label{ineq:sc}
    \end{align}
    
    So, we have shown that $f_v^N$ is close to $(1+\epsilon) q_v^N$. Note that in step 2 of the non-crucial edge procedure, $s_e$ only decreases when $\tilde{x_v}$ is greater than $\max\{q_v^N, \epsilon\}$ . Therefore, we will show that the probability of this happening is small, but we will use $f_v^N$ in place of $q_v^N$.  To do so, we define $X_1, X_2, ..., X_k$ to be the random variables conditioned on edge $e$ already being realized where we set $X_i$ to 0 if $e_i$ is not realized and $(\min\{f_{e_i}/(p_v p_e), 2\tau/(p_v p_e)\})$ if $e_i$ is realized. Note that we conditioned on the realization of edge $e$ to keep the set of events $X_i$ mutually independent. Since edge $e$ is realized, we know that the incident vertex $v$ to each edge $e_i$ must also be realized. Furthermore, the probability that each edge $e_i$ is realized is now only $p_v p_e$ since a factor of $p_v$ has been removed for the incident vertex. Now, we define $X := \sum_{e_i} X_i$.  Observe that $X = \tilde{x_v} - p_v \cdot \tilde{x_e}$ since step 1 of the non-crucial edge procedure sets any realized edges $e_i$ to be $\min\{f_{e_i}/(p_v^2p_e), 2\tau/(p_v^2p_e)\}$.  Additionally by linearity of expectation, $\mathbb{E}[X] = \sum_{i} \mathbb{E}[X_i] \leq f_v$.  
    
    Now we will consider two cases.  First, assume $\mathbb{E}[X] \leq \epsilon/(2p_v)$, then by definition of $X$ we also have that $\mathbb{E}[X] \leq \max\{f_v^N, \epsilon\}/(2p_v)$. From here, we can achieve the following probability bound:
    \begin{align*}
        &\Pr(X > \max\{f_v^N, \epsilon\}) \\ &= \Pr(X - \max\{f_v^N, \epsilon\}/2 \geq \max\{f_v^N, \epsilon\}/2) \\
        &\leq \Pr(X - \mathbb{E}[X] \geq \max\{f_v^N, \epsilon\}/2) \\
        &= \Pr\left(X \geq \left(1 + \frac{\max\{f_v^N, \epsilon\}}{2\cdot\mathbb{E}[X]}\right) \mathbb{E}[X]\right) \\
        &\leq \exp\left(-\frac{\max\{f_v^N, \epsilon\}/2}{6\tau/(p_vp_e)}\right) \text{, by Chernoff bound.} \\
        &\leq \exp\left(-\frac{\epsilon}{12\tau/(p_vp_e)}\right)
        \leq \exp(-1/\epsilon^2)
        \leq \epsilon.
    \end{align*}
    
    For the second case, if $\mathbb{E}[X] > \epsilon/(2p_v)$:
    
    \begin{align*}
        &\Pr(X > (1+\epsilon)\max\{f_v^N, \epsilon\}/p_v) \leq \Pr(X > (1+\epsilon)\mathbb{E}[X]) \\
        &\leq \exp\left(-\frac{\epsilon^2\mathbb{E}[X]}{6\tau/(p_vp_e)}\right) \text{, by Chernoff bound.} \\
        &\leq \exp\left(\frac{\epsilon^3}{12\tau/p_e}\right) \text{, since $\mathbb{E}[X] > \epsilon/(2p_v)$.} \\
        &\leq \exp\left(-\frac{1}{\log(1/\epsilon)}\right) \leq \epsilon.
    \end{align*}
    Thus, in either case we know that 
    \[
    \Pr(X \geq (1+\epsilon) \max\{f_v^N, \epsilon\}) \leq \epsilon. 
    \]  Since $X = \tilde{x_v} - \tilde{x_e}$, we also know that with probability at least $1-\epsilon$,
    \begin{align*}
        \tilde{x_v} &\leq (1+\epsilon)(\max\{f_v^N, \epsilon\}/p_v) +p_v\cdot\tilde{x_e} \\
        &\leq (1+\epsilon)(\max\{f_v^N, \epsilon\}/p_v) + \epsilon \\
        &\leq (1+2\epsilon)(\max\{f_v^N, \epsilon\}/p_v).
    \end{align*}
    Now, we can use Inequality $\ref{ineq:sc}$ to obtain our inequality to substitute in $q_v^N$.  With probability at least $1-2\epsilon$,
    \begin{align*}
        \tilde{x_v} &\leq (1+2\epsilon)(\max\{f_v^N, \epsilon\}/p_v) \\
        &\leq (1+2\epsilon)(1+\epsilon)(\max\{q_v^N, \epsilon\}/p_v) \\
        &\leq (1+5\epsilon)(\max\{q_v^N, \epsilon\}/p_v).
    \end{align*}
    Finally, with probability at least $1-2\epsilon$,
    \begin{align*}
        \max\{q_v^N, \epsilon\}/(\tilde{x}_v \cdot p_v) \geq \frac{1}{1+5\epsilon} \geq 1-5\epsilon.
    \end{align*}
\end{proof}

\begin{lemma}
\label{lem:nbound}
    Given a non-crucial edge matching $x$ from Procedure 1,
    \[
        \mathbb{E}\left[\sum_{e \in N} w_e \cdot x_e \right] \geq (1 - 10 \epsilon) \varphi(N).
    \]
\end{lemma}

\section{Crucial Edges and Unweighted Approximation}
\label{sec:unweighted}
Now we will augment the previous fractional matching from the non-crucial edge procedure with a second procedure on crucial edges. The crucial edge procedure is different for unweighted and weighted graphs, and we first give the slightly simpler procedure for unweighted graphs. The resulting fractional matching will give us our desired approximation ratio to prove Theorem \ref{thm:unweighted}. 

In our new stochastic matching model, it is important to emphasize that the realization of non-crucial edges actually gives some information about the realization of crucial edges.  In fact, because some non-crucial edges share an incident vertex with crucial edges, the probability of a crucial edge being realized loses a factor of $p_v$ when we know an adjacent non-crucial edge already has been realized as that implies the shared incident vertex has been realized as well. In order to avoid this further complexity to the problem, this crucial edge procedure will assign remaining budget to a crucial edge based on the incident vertices' non-crucial edge budget such that no vertex contributes more than size 1 to the whole matching. Specifically, we will set $x_e$ of edge $e = (u, v)$ for crucial edges equal to approximately $1 - q_v^N$ or $1 - q_u^N$ (with an additional $(1-\epsilon)$ factor) in the crucial edge procedure.  Additionally, we utilize a probability distribution defined below.

\begin{definition}
   Take any matching $M$ of the realized crucial edges in sub-graph $Q = (V, E_Q)$ from Algorithm $\ref{alg:nonadaptive1}$. Given $\mathcal{G}$, the $\emph{appearance probability}$ of $M$ is the probability that $M$ is the exact set of crucial edges that is in both $E_Q$ and the maximum matching of $\mathcal{G}$.
\end{definition}

\subsection{Crucial Edge Procedure}
\begin{enumerate}
    \item Draw a matching $M_C$ of $\mathcal{E}_Q \cap C$ based on the appearance probabilities of matchings given $\mathcal{G} \cap C$.
    \item For $e = (u, v) \in M_C$, let $x_e = (1-\epsilon) \min\{1 - q_u^N, 1-q^N_v\}$.
\end{enumerate}

Trivially, this procedure makes sure that all vertices $v \in V$ do not exceed their budget, i.e. $x_v \leq 1$.  So, first, we will prove that this procedure still allows fractional matching $x$ to satisfy the constraints of Lemma $\ref{lem:frac2int}$.

\begin{lemma}
\label{lem:blossom}
    $\forall U \subseteq V$ with $|U| \leq 1 / \epsilon$, $\sum_{e \in E} x_e \leq \lfloor |U|/2 \rfloor$.
\end{lemma}

Second, we show that the size of the fractional matching after the crucial edge procedure is the size we wanted. The proof of the following lemma largely relies on the optimal matching on non-crucial edges from the non-crucial edge procedures as well as a bit of algebra involving matching probabilities and the additional value assignments from the crucial edge procedure.


\begin{lemma}
\label{lem:crucsize}
    Given unweighted graph $G$ and fractional matching $x$ from the non-crucial and crucial edge procedures, $\mathbb{E}[\sum_{e \in \mathcal{E}_Q} x_e]/\mathbb{E}[\mu(G)] \geq (1 - 2\epsilon)(4\sqrt{2}-5)$.
\end{lemma}

Finally, we may prove our first theorem, Theorem $\ref{thm:unweighted}$.

\begin{proof} [Proof of Theorem \ref{thm:unweighted}]
Take fractional matching $x$ on $\mathcal{Q} = (\mathcal{V}, \mathcal{E}_Q)$ as constructed by the non-crucial edge and crucial edge procedures. By Lemma \ref{lem:blossom}, $x$ satisfies the condition for Lemma \ref{lem:frac2int}.  Thus, by Lemma \ref{lem:frac2int}, there exists an integral matching $y$ on $\mathcal{Q}$ with size at least $(1-\epsilon)$ times that of $x$.  So, we have that
\begin{align*}
    &\mathbb{E}[\mu(\mathcal{Q})] \geq \mathbb{E}\left[\sum_{e \in \mathcal{E}_Q} w_e \cdot y_e\right] \\
    &\geq (1-\epsilon)\mathbb{E}\left[\sum_{e \in \mathcal{E}_Q} w_e \cdot x_e\right] &\text{By Lemma \ref{lem:frac2int}.} \\
    &\geq (1-3\epsilon)(4\sqrt{2} - 5)\mathbb{E}[\mu(G)] & \text{By Lemma \ref{lem:crucsize}.} \\
    &\geq (.6568-\epsilon_0)\mathbb{E}[\mu(G)]
\end{align*}
where $\epsilon_0 = 3(.6568)\epsilon$.
\end{proof}
Unfortunately, it can be shown that this bound is tight for our analysis using the non-crucial and crucial edge procedures as there are examples of graphs where the best approximation using these procedures has ratio $(4\sqrt{2} - 5)$ (see \cite{DBLP:conf/soda/BehnezhadFHR19})

\section{Weighted Approximation}
\label{sec:weighted}
For weighted graphs, we will begin similarly to the unweighted graph analysis.  The end goal will be to create a fractional matching of non-crucial edges and crucial edges to prove our desired .501 bound.  We have slightly changed and improved the results of this section compared to \cite{DBLP:conf/soda/BehnezhadFHR19} by working with a smaller constant $\delta$ in Definition \ref{def:heavy1}. From this update, the approximation ratio is a bit better for the following weighted case analysis.

Due the new model, we must again work around the correlation between realization probabilities of adjacent edges in our new model. First, the non-crucial edge procedure will create a near-perfect fractional matching $x$ on non-crucial edges as before. However, a new procedure for crucial edges will be used that will not only assign fractional matching values to crucial edges but may also modify the values in the fractional matching given to non-crucial edges.  By lowering the previously assigned non-crucial edge matching values, there will be more budget for crucial edges to take. In weighted graphs, it is possible that crucial edges have high weights and significantly contribute to the maximum matching of $\mathcal{G}$. For this reason, increasing the remaining vertex budget accordingly for crucial edges is necessary.  So, the following procedure considers matching probabilities and edge weights when assigning fractional matching values to the crucial edges, and then it modifies the non-crucial fractional matching if needed.

\subsection{Weighted Crucial Edge Procedure}
\begin{enumerate}
    \item As in the previous crucial edge procedure, draw some matching $M_C$ of $\mathcal{E}_Q \cap C$ according to appearance-probabilities.  For vertex $v$ and constant $\alpha$, define 
    \[g(v, \alpha) := \frac{\min\{q_v^N, 1- \alpha\}}{q_v^N} \varphi_v^N.\]
    \item For $e = (u, v) \in M_C$, set \[
    x_e := (1-\epsilon) \argmax\limits_{0 \leq \alpha \leq 1} (g(u, \alpha) + g(v, \alpha) + \alpha\cdot w_e).
    \]
    \item For any vertex $v$, if $x_v > 1$, scale down fractional matching of incident non-crucial edges until $x_v \leq 1$.
\end{enumerate}

Before analyzing this new procedure, we want to bound $\varphi(N)$ and $\varphi(C)$ so that they are easier to work with.  To do so, we utilize the following lemma discussing the relation of $\varphi(C)$ to Algorithm $\ref{alg:nonadaptive1}$'s expected matching size. The proof of the lemma and other missing proofs of this section can be found in Appendix.
\begin{lemma}
\label{lem:cbound}
    Given sub-graph $Q = (V, E_Q)$ from Algorithm $\ref{alg:nonadaptive1}$, 
    $\mathbb{E}[\mu(\mathcal{Q})]\geq (1-\epsilon)\varphi(C)$.
\end{lemma}

Recall that from Procedure $\ref{proc:non-crucial}$ and Lemma $\ref{lem:nbound}$, there is an expected matching of $Q$ of size at least $(1-10\epsilon)\varphi(N)$.  From Lemma $\ref{lem:cbound}$, we know that there is also an expected matching size of $Q$ with size at least $(1-\epsilon)\varphi(C)$.  Thus, if either $\varphi(C)$ or $\varphi(N)$ is at least $.501 \cdot \mathbb{E}[\mu(\mathcal{G})]$, we have our approximation ratio from Theorem $\ref{thm:weighted}$. We may now focus on when this is not the case.  Since $\mathbb{E} [\mu(\mathcal{G})] = \varphi(C) + \varphi(N)$, we will analyze the case when
\[
    .499 \cdot \mathbb{E}[\mu(\mathcal{G})] \leq \varphi(C), \varphi(N) \leq .501 \cdot \mathbb{E}[\mu(\mathcal{G})].
\]
For this analysis, we will utilize the matching constructed by the weighted crucial edge procedure and further classify crucial edges by edge weights to make claims about this matching's expected weight.

\begin{definition}
\label{def:heavy1}
    Let $\delta = .09$, then a crucial edge $e = (u,v) \in C$ is $\emph{heavy}$ if $w_e \geq (1 + \delta)(\varphi^N_u + \varphi^N_v)$, and we denote the set of heavy edges with $H$.  
    
    A crucial edge is $\emph{semi-heavy}$ if $e$ is not heavy, $w_e \geq 2(1 + \delta)\varphi_v^N$, and $q_u^N \leq (1-\delta)$ where $q_v^N \geq q_u^N$.  We denote the set of semi-heavy edges with $H^*$.
\end{definition}

Observe that by definition, the weight of a heavy edge $e = (u, v) \in H$ is larger than the expected fractional matching of adjacent non-crucial edges.  As such, from step 2 of the weighted crucial edge procedure, $x_e$ will be maximized when $\alpha = 0$ and $x_e = (1-\epsilon)$.  Furthermore, step 3 will reduce the fractional matching of adjacent non-crucial edges to 0 so that $x_u, x_v \leq 1$.  Similarly, for semi-heavy edge $e = (u, v) \in H^*$, by definition the weight of $e$ is greater than the expected non-crucial fractional matching of one of its vertices.  Specifically, from step 2 of the procedure we have that $x_e \geq (1-\epsilon)(1-q_u^N) \geq (1-\epsilon)\delta$.  With these crucial edge bounds, we can prove that if a graph has enough heavy and semi-heavy edges, then will will reach at least the desired .501 approximation ratio.

\begin{lemma}
\label{lem:weightedsupport}
    If $\varphi(H) + \varphi(H*) \geq 0.074\varphi(C)$, then $\mathbb{E}[\mu(\mathcal{Q})]/\mathbb{E}[\mu(\mathcal{G})] \geq .501 - 11\epsilon$.
\end{lemma}

Now, we must look at when $\varphi(H) + \varphi(H*) < 0.074\varphi(C)$, i.e. when the heavy and semi-heavy edges do not constitute a large part of the expected maximum matching on crucial edges. We will use $C^*$ to denote the set of crucial edges that are not heavy nor semi-heavy.  We will \emph{direct} edges of $C^*$ based on their adjacent non-crucial edges' contribution to the maximum matching of $\mathcal{G}$.

\begin{definition}
  We will classify edges $e = (u, v) \in C^*$ into the following three types of edges:
\begin{enumerate}
    \item If $\varphi_v^N \geq \varphi_u^N$, direct $e$ towards $v$.
    \item If $\varphi_v^N < \varphi_u^N$ and $w_e \leq 2(1+\delta)\varphi_v^N$, direct $e$ towards $v$.
    \item Else, direct $e$ towards $u$.  Note, in this case $\varphi_v^N < \varphi_u^N$ and $w_e > 2(1+\delta)\varphi_v^N$.
\end{enumerate}
\end{definition}

Note that edge types 1-3 partition $C^*$.  Thus, using these definitions and directed edge types, we may prove properties of $C^*$ and our desired bound for the approximation ratio of sub-graph $Q$ from Algorithm $\ref{alg:nonadaptive1}$.





\begin{theorem}
\label{thm:mainweightednaive}
    Given sub-graph $Q$ from Algorithm $\ref{alg:nonadaptive1}$, $\mathbb{E}[\mu(\mathcal{Q})]/\mathbb{E}[\mu(\mathcal{G})] \geq .501 - 11\epsilon$.
\end{theorem}


\section{EDCS 2/3 Approximation}
\label{sec:edcs}
In this section we improve our bound for unweighted graphs to $2/3-\epsilon$ for an arbitrary constant $\epsilon>0$. Inspired by a work of \cite{DBLP:conf/soda/AssadiB19} we use edge-degree constrained sub-graph (EDCS) for designing our algorithm. Before stating our result we first give the definition of EDCS from \cite{DBLP:conf/icalp/BernsteinS15} and \cite{DBLP:conf/soda/BernsteinS16}. 

\begin{definition}
[\cite{DBLP:conf/icalp/BernsteinS15,DBLP:conf/soda/BernsteinS16}]
\label{def:edcs}
For any graph $G=(V,E)$, and integers $\beta \ge \beta^- \ge 0$, an edge-degree constrained sub-graph(EDCS)$(G,\beta, \beta^-)$ is a sub-graph $H=(V, E_H)$ with the following two properties.
\begin{enumerate}
    \item For every edge $(v,u) \in E_H$: $deg_H(v)+ deg_H(u) \le \beta$.
    \label{enum:mainedcs1}
    \item For every edge $(v,u) \notin E_H$: $deg_H(v)+ deg_H(u) \ge \beta^-$.
    \label{enum:mainedcs2}
\end{enumerate}
\end{definition}

It has been shown in \cite{DBLP:conf/icalp/BernsteinS15} and \cite{DBLP:conf/soda/BernsteinS16} that for any graph $G$, and any parameters $\beta > \beta^-$, an EDCS of $G$ exists. Also it is easy to see that an EDCS of $G$ is degree-bounded and has a maximum degree of $\beta$. An interesting property of EDCS is that for a large enough $\beta$ and $\beta^-$, it always preserves $2/3$ approximation of maximum matching in $G$. Specifically, we have the following.

\begin{theorem} [\cite{DBLP:conf/soda/AssadiB19}]
\label{thm:edcs}
Let $G=(V,E)$ be any graph, $\epsilon <1/2$, $\lambda \le \epsilon/32$, $\beta \ge 8 \lambda^{-2} \log(1/\lambda)$, $\beta^- \ge (1-\lambda) \beta$, and Let $H$ be an EDCS$(G,\beta, \beta^-)$. Then $\mu(H)\ge (2/3 - \epsilon) \mu(G)$. 
\end{theorem}

A result by \cite{DBLP:conf/soda/AssadiB19} shows that for any stochastic graph $G$ where each edge is realized with a probability of $p_e $, an EDCS$(G, \beta, \beta-1)$ also preserves a $2/3-\epsilon$ approximation of the expected maximum matching. We show a similar result for the generalized stochastic matching problem where both edges and vertices are stochastic. Specifically, let $Q$ be an EDCS$(G,\beta, \beta-1)$ for $\beta \ge \frac{C \log(1/(\epsilon \cdot  p_v \cdot p_e))}{\epsilon^2 p_v p_e}$ where $C$ is a large constant. 
Also, Let $\mathcal{Q}$ be the realized portion of $Q$. In the following lemma we show that using the EDCS approach, $\mathcal{Q}$ achieves a $2/3 - O(\epsilon)$ matching approximation ratio in expectation. The proof of Lemma \ref{lem:edcs-main} can be seen in the final section of the Appendix. 

\begin{lemma}
\label{lem:edcs-main}
$\mathbb{E}[\mu(\mathcal{Q})] \ge (2/3-O(\epsilon)) \mathbb{E}[\mu(\mathcal{G})]$ \,.
\end{lemma}

\section{Conclusion}
We have now proven bounds on the approximation ratio of Algorithm \ref{alg:nonadaptive1} breaking a half-approximation.  The natural next step is to improve the ratio up to $(1-\epsilon)$.  For the old model of stochastic matching, \cite{DBLP:journals/corr/abs-2004-08703} achieved a $(1-\epsilon)$ approximation for weighted stochastic matching using \ref{alg:nonadaptive1} complemented by a greedy sub-algorithm. Unfortunately, differences in the models bar us from adapting their new algorithm directly. Previous analysis techniques relied on the complete independence of edge realization, and it seems to us to be non-trivial to overcome this difference.

\medskip

\section{Acknowledgements}
\label{sec:ack}
This research was supported by the NSF BIGDATA Grant No. 1546108, NSF SPX Grant No. 1822738, NSF AF Grant No. 2114269, and an Amazon AWS award.

\bibliography{refs}

\clearpage

\appendix

\section{Missing Proofs of Algorithm \ref{alg:nonadaptive1} Analysis}
\label{appx:algana}
\begin{proof} [Proof of Lemma \ref{lem:nsupport}]
     To prove the first property, note that from step 1 of the procedure we have that $x_e \leq 2\tau/(p_v^2p_e) \leq \epsilon^3$ for at most ${|U| \choose 2}$ edges.  This implies that for any $U \subseteq V$ with $|U| \leq 1/\epsilon$,
     \begin{align*}
        \sum_{e \in E(U)} x_e &= \epsilon^3 \frac{|U| (|U| - 1)}{2}
        \leq \epsilon^3  \frac{\frac{1}{\epsilon}(|U| - 1)}{2} \\
        &\leq \epsilon^2 \lfloor|U| / 2 \rfloor \leq \epsilon \lfloor|U| / 2 \rfloor \,.
     \end{align*}
     To prove the second property, we then observe the scaling from steps 2 and 3 for the edges incident to any vertex $v \in V$.  Since for all edges $s_e \leq 1$, if $\tilde x_v \leq \max\{q^N_v, \epsilon\}$, then trivially $x_v = s_e \cdot \tilde x_v \leq \max\{q^N_v, \epsilon\}$.  
     
     If $\tilde x_v > \max\{q^N_v, \epsilon\}$, then $\max\{q^N_v, \epsilon\} / \tilde x_v \leq 1$ and so, at the end of step 2, $s_e \leq \max\{q^N_v, \epsilon\} / \tilde x_v$.  Thus after step 3,
     \[
        x_v = \tilde x_v \cdot s_e \leq \tilde x_v \cdot \max\{q^N_v, \epsilon\} / \tilde x_v \leq \max\{q^N_v, \epsilon\} \,.
     \]
\end{proof}

\begin{proof} [Proof of Lemma \ref{lem:ns1}]
First, we will show a simple equality between the non-crucial edges in $Q$ and $\varphi(N)$.  Let $\mu_i^N$ denote the random variable representing the weight of the non-crucial edges in the maximum matching from round $i$ of Algorithm $\ref{alg:nonadaptive1}$. In Algorithm $\ref{alg:nonadaptive1}$, the matching chosen at any round $i$ has the same probability of being chosen as the matching chosen for $\mathcal{G}$.  In other words, $\mathbb{E}[\mu_i^N] = \sum_{e \in N} w_e \cdot q_e = \varphi(N)$.  Additionally, let $\overline{\mu^N} = (\mu_1^N + \mu_2^N + ... + \mu_R^N)/R$. Then, by definition of $f_e$, $\overline{\mu^N} = \sum_{e \in E_Q \cap N} w_e \cdot f_e$, and moreover, by linearity of expectation, we have that $\mathbb{E}[\overline{\mu^N}] = \varphi(N)$. Thus, we have the key equality between $\varphi(N)$ and our non-crucial edges of $Q$:
\begin{align}
    \mathbb{E}\left[\sum_{e \in E_Q \cap N} w_e \cdot f_e\right] = \varphi(N).
\end{align}
Unfortunately, this does not tell us enough about the fractional matching on the realization of $Q$, and so, we will need to analyze step 1 of our non-crucial edge procedure. As previously discussed, $\mathbb{E}[f_e] = q_e$ is the observation that was the basis of our choice of procedure.  However, since $x_e = \min\{f_e/(p_v^2p_e), 2\tau/(p_v^2p_e)\}$, we will discuss the probability that $2\tau < f_e$.  We want to show that this probability is small, only at most $\epsilon \cdot q_e$, and thus, does not affect our matching much.  To do so, we let $X = X_1 + X_2 + ... + X_R$ be the sum of random independent events $X_i$, which are 1 if edge $e$ is in the maximum matching of iteration $i$ of Algorithm $\ref{alg:nonadaptive1}$ and 0 otherwise.  We will use a Chernoff bound on $\mathbb{E}[X]$ and work back to $f_e$ and $\tau$ since by definition, $X = R \cdot f_e$.
\begingroup
\allowdisplaybreaks
\begin{align*}
&\Pr(f_e \geq 2 \tau ) = \Pr(f_e - \tau \geq \tau )\\
&\leq P[ f_e - q_e \ge \tau] \text{, since } e \text{ is non-crucial and } q_e < \tau .
\\& = \Pr\Big( f_e - \mathbb{E}[f_e] \ge \tau \Big) 
\\& = \Pr\Big( X - \mathbb{E}[X] \geq R \cdot \tau \Big) \text{, since $X= f_e \cdot R$.}
\\ & \leq \exp\Big(- \frac{R \cdot \tau \cdot \log\big(1+(R \cdot \tau)/\mathbb{E}[X]\big)}{2}\Big), \\ &\text{by Chernoff bound.}
\\ & \leq \exp\Big(-\frac{R \cdot \tau \cdot \log(1+\frac{\tau}{q_e})}{2}\Big), \\ &\text{since $\mathbb{E}[X]= q_e \cdot R$.}
\\ & \leq \exp\Big(-50 \log(1/(\epsilon p_v^2p_e)) \log(1+\frac{\tau}{q_e})\Big) ,\\ &\text{since } R \cdot \tau > 100 \log(1/(\epsilon p_v^2 p_e)).
\\ & = \frac{1}{\exp\Big(50 \log(1/(\epsilon p_v^2p_e)) \log(1+\frac{\tau}{q_e})\Big)}
\\ & = \frac{1}{(1+\frac{\tau}{q_e})^{50 \log(1/(\epsilon p_v^2p_e))}}
\\ & \leq \frac{1}{(1+\frac{\tau}{q_e})(1+\frac{\tau}{q_e})^{49 \log(1/(\epsilon p_v^2p_e))}}, \\ &\text{since $\epsilon \leq e^{-1}$ and therefore $\log(1/(\epsilon p_v^2p_e)) \ge 1$.}
\\ & \leq \frac{1}{(1+\frac{\tau}{q_e})\cdot 2^{49 \log(1/(\epsilon p_v^2p_e))}}
,\\& \text{since $\tau > q_e$ for $e \in N$ and so,  $1+\frac{\tau}{q_e}>2$.}
\\ & = \frac{1}{(1+\frac{\tau}{q_e})\cdot \exp\big({49/\log(2) \cdot \log(1/(\epsilon p_v^2p_e))}\big)}
\\ & \leq \frac{1}{(1+\frac{\tau}{q_e})\cdot \exp\big(30\log(1/(\epsilon p_v^2p_e))\big)}
\\ & \leq \frac{1}{(1+\frac{\tau}{q_e})\cdot e^{10} \cdot \exp\big(20\log(1/(\epsilon p_v^2p_e))\big)}
\\ & \leq \dfrac{1}{20 (1+\frac{\tau}{q_e}) \cdot \exp\big(5 \log(1/(\epsilon p_v^2p_e))\big)}
\\ & \leq \dfrac{1}{20 (1+\frac{\tau}{q_e}) \cdot \log(1/\epsilon) \cdot \exp\big(4 \log(1/(\epsilon p_v^2p_e))\big)}, \\
& \text{since $e^{x} \ge x$ for all real numbers $x$ and $p_v, p_e \leq 1$.}
\\ & = \dfrac{\epsilon^4 p_v^8p_e^4}{20 (1+\frac{\tau}{q_e}) \cdot \log(1/\epsilon)}
\\ & \leq \dfrac{\epsilon \tau} {(1+\frac{\tau}{q_e})} \text{, since } \tau= \frac{\epsilon^3 p_v^2p_e}{20 \log(1/\epsilon)}.
\\ & = \dfrac{\epsilon \cdot \tau \cdot q_e} {\tau+q_e}
\\ & \leq \dfrac{\epsilon \cdot \tau \cdot q_e} {\tau}
\\ & = \epsilon \cdot q_e.
\end{align*}
\endgroup
Now, we can bound the size of $\tilde{x_e}$.  First, note that by conditional expectation we know the following inequality
\begin{align}
\label{ineq:lem2.7.1}
    \mathbb{E}[\min\{f_e, 2\tau\}] \geq \Pr(f_e \leq 2\tau ) \cdot \mathbb{E}[f_e | f_e \leq 2\tau].
\end{align}

Additionally, we have
\begin{align*}
    \mathbb{E}[f_e] = \Pr(f_e \leq 2\tau)\mathbb{E}[f_e | f_e \leq 2\tau] \\+ \Pr (f_e > 2\tau) \mathbb{E}[f_e | f_e > 2\tau].
\end{align*}

Subtracting inequality \ref{ineq:lem2.7.1} from this equation yields

\begin{align*}
    \mathbb{E}[f_e] - \mathbb{E}[\min\{f_e, 2\tau\}] &\leq \Pr(f_e > 2\tau)\mathbb{E}[f_e | f_e > 2\tau].
\end{align*}
We can then substitute the inequality $\Pr(f_e \geq 2\tau) \geq \epsilon \cdot q_e$ proven above to obtain
\begin{align*}
    \mathbb{E}[f_e] - \mathbb{E}[\min\{f_e, 2\tau\}] &\leq (\epsilon\cdot q_e)\mathbb{E}[f_e | f_e > 2\tau] \\ &\leq \epsilon\cdot q_e \\ &= \epsilon \cdot \mathbb{E}[f_e] \\
    \implies  \mathbb{E}[\min\{f_e, 2\tau\}] &\geq (1-\epsilon) \mathbb{E}[f_e].
\end{align*}

Next, we will incorporate edge weights.  Weights are not necessary for this section since we are currently focused on unweighted graphs but will allow us to use this lemma later in the Weighted Approximations section.  As such, we note that
\begin{align}
    \mathbb{E}\left[\sum_{E_Q\cap N} \min\{f_e, 2\tau\} \cdot w_e\right] &= \sum_{E_Q \cap N} \mathbb{E}[\min\{f_e, 2\tau\}] \cdot w_e \nonumber \\ &\geq \sum_{E_Q \cap N} (1-\epsilon)\mathbb{E}[f_e] \cdot w_e \nonumber \\
    &= (1-\epsilon)\sum_{E_Q \cap N} q_e \cdot w_e \nonumber \\
    &= (1-\epsilon) \varphi(N). \label{ineq:s1}
\end{align}

Lastly for step 1, we can now bound the value of $\tilde{x_e}$.  First, we use $\mathbf{1}_e$ to denote the indicator of the event that edge $e$ is realized.  Note that any edge $e$ is realized with probability $p_e$ only if both of its vertices are realized with probability $p_v$, i.e. edge $e$ is realized with probability $p_v^2 p_e$.  Altogether, for step 1 we have
\begin{align*}
    &\mathbb{E}\left[\sum_{e \in \mathcal{E}_Q \cap N} w_e \cdot \tilde{x_e}\right] \\ &\geq \mathbb{E}\left[\sum_{e \in \mathcal{E}_Q \cap N} w_e \cdot \min\{f_e/(p_v^2 p_e), 2\tau/(p_v^2p_e)\}\right] \\
    &= 1/(p_v^2 p_e) \mathbb{E}\left[ \sum_{e \in \mathcal{E}_Q \cap N} w_e \cdot \min\{f_e, 2\tau\}\right] \\
    &= 1/(p_v^2 p_e) \mathbb{E}\left[\sum_{e \in E_Q \cap N} w_e \cdot \min\{f_e, 2\tau\} \cdot \mathbf{1}_e\right].
\end{align*}
Note the switch from realized sub-graph $\mathcal{Q}$ to $Q$ by using the indicator variable $1_e$ for the edges of the sum.  To finish up, we have
\begin{align*}
    &= 1/(p_v^2 p_e) \sum_{e \in E_Q \cap N} \mathbb{E}[w_e \cdot \min\{f_e, 2\tau\} \cdot \mathbf{1}_e] \\
    &= 1/(p_v^2 p_e) \sum_{e \in E_Q \cap N} w_e \cdot \mathbb{E}[  \min\{f_e, 2\tau\}] \cdot 
    \mathbb{E}[\mathbf{1}_e] \\
    &= 1/(p_v^2 p_e) \sum_{e \in E_Q \cap N} w_e \cdot \mathbb{E}[  \min\{f_e, 2\tau\}] \cdot (p_v^2 p_e) \\
    &=  \sum_{e \in E_Q \cap N} w_e \cdot \mathbb{E}[  \min\{f_e, 2\tau\}] \\
    &\geq (1-\epsilon)\varphi(N), \text{ from Inequality $\ref{ineq:s1}$.}
\end{align*}
\end{proof}

\begin{proof} [Proof of Lemma \ref{lem:nbound}]
    First, note that the non-crucial edge procedure only assigns non-zero values to edges in the realized sparse sub-graph $ \mathcal{Q}$.  So, with this fact,
    \begin{align*}
        \mathbb{E}\left[\sum_{e \in N} w_e \cdot x_e\right] &= \mathbb{E}\left[\sum_{e \in \mathcal{E}_Q \cap N} w_e \cdot x_e \right] \\
        &= \mathbb{E}\left[\sum_{e \in \mathcal{E}_Q \cap N} w_e \cdot s_e \cdot \tilde{x_e}  \right].
    \end{align*}
    From Lemma $\ref{lem:ns2}$, we know that for any edge $e \in \mathcal{E}_Q \cap N$, the probability that the scaling factor $s_e$ from steps 2 and 3 of the non-crucial edge procedure is less than $1-5\epsilon$ is at most $2\epsilon$ given a vertex $v \in e$.  Since there are two incident vertices for edge $e$, $s_e \geq 1-5\epsilon$ with probability at least $1-4\epsilon$, and thus,
    \begin{align*}
        &\mathbb{E}\left[\sum_{e \in \mathcal{E}_Q \cap N} w_e \cdot s_e \cdot \tilde{x_e}  \right] \\ &\geq (1-4\epsilon)(1-5\epsilon)\mathbb{E}\left[\sum_{e \in \mathcal{E}_Q \cap N} w_e \cdot \tilde{x_e}  \right] &\text{By Lemma \ref{lem:ns2}.} \\
        &\geq (1-9\epsilon)\mathbb{E}\left[\sum_{e \in \mathcal{E}_Q \cap N} w_e \cdot \tilde{x_e}  \right] \\
        &\geq (1-9\epsilon)(1-\epsilon)\varphi(N) &\text{By Lemma \ref{lem:ns1}.} \\
        &\geq (1-10\epsilon) \varphi(N).
    \end{align*}

\end{proof}

\section{Missing Proofs of Crucial Edges and Unweighted Approximation}
\label{appx:unweighted}

\begin{proof} [Proof of Lemma \ref{lem:blossom}]
From Lemma $\ref{lem:nsupport}$, the non-crucial edge procedure creates a fractional matching with size at most $\epsilon \lfloor \frac{|U| - 1}{2} \rfloor$.  Now, from the crucial edge procedure, since $M_C$ is an integral matching, it can only have at most $\lfloor\frac{|U|-1}{2} \rfloor$ edges with each edge contributing at most $1 - \epsilon$ to the fractional matching $x$ after step 2.  In total, $x$ will have size at most
\[
    \epsilon \lfloor \frac{|U| - 1}{2} \rfloor + (1 - \epsilon) \lfloor \frac{|U| - 1}{2} \rfloor =  \lfloor \frac{|U| - 1}{2} \rfloor
\]
\end{proof}

\begin{proof} [Proof of Lemma \ref{lem:crucsize}]
    Since the crucial edge procedure builds on the fractional matching of the non-crucial edge procedure, we have that $\sum_{e \in E} x_e = \sum_{e \in N} x_e + \sum_{e \in C} x_e$.  Then, from Lemma $\ref{lem:nbound}$, we have that $\sum_{e \in N} x_e\geq (1-\epsilon) \varphi(N) = (1-\epsilon)q(N)$ since $w_e = 1$ for unweighted graphs.  So, we will be focused on showing the additional size from the crucial edges.
    
    Given some crucial edge $e$, in order for $x_e$ to be assigned a nonzero value by the crucial edge, the edge must be included in $Q$ by Algorithm $\ref{alg:nonadaptive1}$ and then in the chosen matching $\mu^C$ given that $e$ is in sub-graph $Q$.  Since $e$ is a crucial edge, by definition it is in $Q$ with probability greater than $1-\epsilon$.  Given that $e$ is in $Q$, some $\mu^C$ such that $e \in \mu^C$ will be chosen by the crucial edge procedure in step 1 with probability at least $q_e$ since $e$ appears in the maximum matching of $\mathcal{G}$ with probability $q_e$ and has already been chosen by $Q$.   The value given to $x_e$ in step 2 will be $(1-\epsilon) \min\{1 - q_u^N, 1 - q_v^N\}$, and so,
    \begin{align}
    \label{ineq:lem2.12.1}
        \mathbb{E}[x_e] &= (1-\epsilon)q_e(1-\epsilon)\min\{1-q_u^N, 1-q_v^N\} \nonumber \\
        &\geq (1-2\epsilon)q_e \min\{1-q_u^N, 1-q_v^N\}.
    \end{align}
    
    To simplify the inequality and remove the min function, the crucial edges will be directed towards the endpoint with the lower budget remaining prior to procedure 2.  In other words, $e = (u, v) \in C$ will be directed to $u$ such that $q_u^N > q_v^N$, ties decided arbitarily.  Let $I_v$ denote the set of incoming crucial edges to vertex $v$ and let $q^I_v := \sum_{e \in I_v} q_e$ be the matching probability of the edges directed towards $v$.  Now, utilizing these definitions with inequality \ref{ineq:lem2.12.1} yields
    \begin{align}
        \label{ineq:lem2.12.2}
        \mathbb{E}\left[\sum_{e \in C} x_e\right] &\geq \sum_v (1-2\epsilon)(1-q_v^N)q_v^I \nonumber \\
        &= (1-2\epsilon)\sum_v(q_v^I - q_v^Nq_v^I) \nonumber \\
        &= (1-2\epsilon)\sum_v(q_v^I - q_v^Nq_v^I) \nonumber \\
        &= (1-2\epsilon)q(C) - (1-2\epsilon) \sum_{v \in V} q_v^N q_v^I.
    \end{align}
    Then for the entire fractional matching, combining Inequality \ref{ineq:lem2.12.2} with Lemma \ref{lem:nbound} we have
    \begin{align*}
        \mathbb{E}\left[\sum_{e \in E} x_e\right] &= \mathbb{E}\left[\sum_{e\in N}x_e\right] + \mathbb{E}\left[\sum_{e \in C} x_e\right] \\
        &\geq (1-2\epsilon)\left( q(N) + q(C) -\sum_{v \in V} q_v^N q_v^I\right).
    \end{align*}
    Finally, we will utilize this equation to bound the approximation ratio. Observe that for unweighted graphs, $\mathbb{E}[\mu(\mathcal{G})] = q(N) + q(C)$.  Furthermore, since each crucial edge is directed, $q(C) = \sum_{v \in V} q_v^I$, and similarly, $2q(N) = \sum_{v \in V} q_v^N$ since edge will be counted twice, once per incident vertex.  With these substitutions, the approximation ratio is as follows
    \begin{align}
    \label{ineq:lem2.12.3}
        \frac{\mathbb{E}[\sum_{e \in E} x_e]}{\mathbb{E}[\mu(\mathcal{G})]} &\geq (1-2\epsilon)\frac{\left( q(N) + q(C) -\sum_{v \in V} q_v^N q_v^I\right)}{q(N) + q(C)} \nonumber \\
        &\geq (1-2\epsilon)\left(1 - \frac{\sum_{v \in V} q_v^N q_v^I}{\sum_{v \in V} q_v^I + \frac{q_v^N}{2}}\right).
    \end{align}
    
    To finish the proof off, we will need one algebraic lemma, which is written and proven as Lemma 5.4 in \cite{DBLP:conf/soda/BehnezhadFHR19}.

\begin{lemma}
\label{lem:algebraic}
Given sets of numbers $a_1, ..., a_n$, $b_1, ..., b_n$ such that
\begin{itemize}
    \item $\forall i \in [1, n]$, $a_i \geq 0, b_i \geq 0$, and $a_i + b_i \leq 1$
    \item $\sum_{i=1}^n a_i + b_i > 0$
\end{itemize}
then it is true that $\frac{\sum_{i=1}^n a_i b_i}{\sum_{i=1}^n a_i + \frac{b_i}{2}} \leq 6 - 4\sqrt{2}$.
\end{lemma}

    Utilizing Lemma \ref{lem:algebraic} and applying to Inequality \ref{ineq:lem2.12.3}
    we have
    \begin{align*}
        \frac{\mathbb{E}[\sum_{e \in E} x_e]}{\mathbb{E}[\mu(\mathcal{G})]} &\geq (1-2\epsilon)(1 - (6 - 4\sqrt{2})) \\
        &\geq (1-2\epsilon)(4\sqrt{2} - 5).
    \end{align*}
    In the non-crucial and crucial edge procedures, we only assign fractional matching values to the edges of $\mathcal{Q}$. Thus, $\mathbb{E}\left[\sum_{e \in E} x_e\right] = \mathbb{E}\left[\sum_{e \in \mathcal{E}_Q} x_e\right]$, and we are finished.
\end{proof}

\section{Missing Proofs of Weighted Approximation}
\label{appx:weighted}
\begin{proof} [Proof of Lemma \ref{lem:cbound}]
    Given edge $e \in C$, let $r_e$ be the probability $e$ is included in $Q$ by Algorithm $\ref{alg:nonadaptive1}$.  Note that by construction of $Q$, $1-r_e = (1-q_e)^R$, and since $e \in C$, we have $1-r_e \leq (1-\tau)^R$.  Furthermore, observe that $R > \frac{ \log(1/\epsilon)}{\tau}$, and so, 
    \begin{align}
        \label{ineq:lem3.1.1}
        1-r_e &\leq ((1-\tau)^{(1/\tau)})^{\log(1/\epsilon)} \nonumber \\
         &\leq (\frac{1}{e})^{\log(1/\epsilon)} = \epsilon \nonumber \\
         \implies r_e &\geq 1-\epsilon.
    \end{align}
    Now, let $\mathbf{1}_Q(e) := 1$ if $e \in \mathcal{E}_Q$, 0 otherwise.  With $r_e \geq 1- \epsilon$, we can finish the proof of the lemma using Inequality \ref{ineq:lem3.1.1} in the fifth line.
    \begin{align*}
        \mathbb{E}[\mu(\mathcal{Q})] &\geq 
        \mathbb{E}[\varphi(E_Q \cap C)] \\ &= \mathbb{E}\left[\sum_{e \in C} w_e \cdot q_e \cdot \mathbf{1}_Q(e)\right] \\
        &= \sum_{e \in C} w_e \cdot q_e \cdot \mathbb{E}[\mathbf{1}_Q(e)] \\
        &= \sum_{e \in C} w_e \cdot q_e \cdot r_e \\
        &\geq (1-\epsilon) \sum_{e \in C} w_e \cdot q_e \\
        &= (1-\epsilon) \varphi(C).
    \end{align*}
\end{proof}

\begin{lemma}
\label{lem:weightedblossom}
    Given fractional matching $x$ from the weighted crucial edge procedure.  $\forall U \subseteq V$ with $|U| \leq 1 / \epsilon$, $\sum_{e \in E} x_e \leq \lfloor |U|/2 \rfloor$.
\end{lemma}

\begin{proof}
    Proof is the same as proof of Lemma \ref{lem:blossom} by substituting the weighted crucial edge procedure in for the crucial edge procedure.
\end{proof}

\begin{proof} [Proof of Lemma \ref{lem:weightedsupport}]
    As noted in the paragraphs preceding Lemma 3.3, by definition of $H$, any heavy edge $e =(u, v) \in H$ will contribute to the fractional matching $x$ originally from the non-crucial edge procedure with an amount at least $(1-\epsilon)w_e - (\varphi_u^N + \varphi_v^N)$.  Again by definition of heavy edges, $w_e/(1 + \delta) \geq (\varphi_u^N + \varphi_v^N)$, which implies that 
    \begin{align}
        (1-\epsilon)w_e - (\varphi_u^N + \varphi_v^N) &\geq (1-\epsilon)w_e- \frac{1}{1+\delta}w_e \nonumber \\&= (\frac{\delta}{1+\delta} - \epsilon)w_e. \label{ineq:heavy}
    \end{align}
    Similarly, from the weighted crucial edge procedure, any semi-heavy edge $e = (u, v) \in H^*$ will contribute
    \begin{align*}
        &(1-\epsilon)(1-q_u^N)w_e - (q_v^N-q_u^N)\varphi_v^N \\ &\geq  (1-\epsilon)(1-q_u^N)w_e - (1-q_u^N)\varphi_v^N \\&= (1-q_u^N)((1-\epsilon)w_e - \varphi_v^N).
    \end{align*}
    By definition of semi-heavy edges, $w_e \geq 2(1+\delta)\varphi_v^N$ and $1-q_u^N \geq \delta$, and so,
    \begin{align}
        &(1-q_u^N)((1-\epsilon)w_e - \varphi_v^N) \nonumber \\ &\geq 
        (1-q_u^N)((1-\epsilon)w_e - \frac{1}{2(1+\delta)}w_e) \nonumber \\
        &= (1-q_u^N)(\frac{1+2\delta}{2(1+\delta} - \epsilon)w_e \nonumber \\ &\geq \delta(\frac{1+2\delta}{2(1+\delta} - \epsilon)w_e \nonumber \\
        &\geq (\frac{\delta+2\delta^2}{2(1+\delta)} -\epsilon)w_e. \label{ineq:semi}
    \end{align}
    From inequalities $\ref{ineq:heavy}$ and $\ref{ineq:semi}$, the total weight of the expected matching will be at least
    \[
        (1-10\epsilon)\varphi(N) + (\frac{\delta}{1+\delta} - \epsilon)\varphi(H) + (\frac{\delta+2\delta^2}{2(1+\delta)} - \epsilon)\varphi(H^*).
    \]
    Since $\delta = .09$, note that $\frac{\delta}{1+\delta} \geq 0.048$ and $\frac{\delta+2\delta^2}{2(1+\delta)} \geq 0.048$.  Then, the fractional matching has an expected weight of at least
    \begin{align*}
        		&(1-10\epsilon) \varphi(N) + (0.048-\epsilon) (\varphi(H)+\varphi(H^*))
		\\& \ge (1-10\epsilon) \varphi(N) + (0.048-\epsilon) (0.09 \varphi(C))
		\\& = (1-10\epsilon) (\mathbb{E}[\mu(\mathcal{G})] - \varphi(C)) + (0.00432 - \epsilon) \varphi(C) \\& \text{since } \varphi(N)+\varphi(C) = \mathbb{E}[\mu(\mathcal{G}].
		\\& \ge (1-10\epsilon) \mathbb{E}[\mu(\mathcal{G}] - \varphi(C) (1-0.0.00432)
		\\& \ge (1-10\epsilon) \mathbb{E}[\mu(\mathcal{G}] - 0.5011 \cdot \mathbb{E}[\mu(\mathcal{G}] (1-0.00432)
		\\& \text{since }\varphi(C) \le .05011 \cdot \mathbb{E}[\mu(\mathcal{G}].
		\\&\ge (0.50106-10\epsilon) \cdot \mathbb{E}[\mu(\mathcal{G}] .
    \end{align*}
    Thus, our fractional matching has the desired approximation ratio. Note that fractional matching $x$ satisfies the requirements of Lemma \ref{lem:frac2int} by Lemma \ref{lem:weightedblossom}.  When converting the fractional matching to an integral matching, by Lemma $\ref{lem:frac2int}$ we introduce another factor of $(1-\epsilon)$, and so, we will end up with an integral matching with approximation ratio $(0.501-11\epsilon)$.
\end{proof}

\begin{lemma}
\label{lem:weightedw}
   Let $e = (u, v) \in C^*$ be directed towards $v$.  Then, $w_e \leq 2(1+\delta)\varphi_v^N$.
\end{lemma}
\begin{proof} [Proof of Lemma \ref{lem:weightedw}]
    Let edge $e = (u, v)$ be a type 1 edge directed towards $v$; by definition we have $\varphi_v^N \geq \varphi_u^N$.  Additionally, $e \in C^*$ must not be a heavy edge, so $w_e \leq (1+\delta)(\varphi_u^N+\varphi_v^N) \leq 2(1+\delta)\varphi_v^N$.
    
    Now, if $e$ is a type 2 edge, then by definition $w_e \leq 2(1+\delta)\varphi_v^N$.  Finally, let $e = (u, v)$ be a type 3 edge directed towards $v$.  Then, by definition of type 3 edges, $\varphi_v^N > \varphi_u^N$, and similar to type 1 edges, $e \in C^*$ must not be a heavy edge.  Therefore by definition of heavy edges, $w_e \leq (1+\delta)(\varphi_u^N+\varphi_v^N) \leq 2(1+\delta)\varphi_v^N$.
\end{proof}

\begin{lemma}
\label{lem:weightedq}
    Let $e = (u, v) \in C^*$ with $q_v^N \geq q_u^N$.  If $e$ is directed towards $u$, then $q_v^N \leq q_u^N + \delta$.
\end{lemma}
\begin{proof} [Proof of Lemma \ref{lem:weightedq}]
    Given edge $e = (u, v) \in C^*$ with $q_v^N \geq q_u^N$ directed towards $u$, $e$ must be type 3 by definition.  So, 
    \[w_e > 2(1+\delta)\varphi_v^N = 2(1+\delta)q_v^N. \]
      Also, since $e$ is not semi-heavy, $q_u^N > (1-\delta)$.  Combining these two inequalities yields the result 
      \[
      q_v^N \leq 1 < q_u^N + \delta.
      \]
\end{proof}

\begin{proof} [Proof of Theorem \ref{thm:mainweightednaive}]
    Given fractional matching $x$ from the weighted crucial edge procedure on $Q$, we will bound the size of $x$ to prove the desired approximation ratio.  Additionally, if $\varphi(H) + \varphi(H^*) \geq 0.074\varphi(C)$, then we are done from $\ref{lem:weightedsupport}$.  Thus, we will work with the case when $\varphi(C^*) > 0.926 \varphi(C)$.
    
    Now, let $I_v$ be the set of incoming crucial edges in $C^*$ directed towards vertex $v$, and similarly, $\varphi_v^I := \sum_{e \in I_v} \varphi_e$. After the non-crucial edge procedure, an edge may have a remaining budget of $(1-\epsilon)(1-\max\{q_u^N, q_v^N\})$.  If $e$ is directed towards $v$, then by Lemma $\ref{lem:weightedq}$ the remaining budget is at least $(1-\epsilon)(1-\delta - q_v^N)$.  Also, recall that the probability that any crucial edge $e$ is included in $Q$ by Algorithm $\ref{alg:nonadaptive1}$ is at least $(1-\epsilon)$ by definition.  So,
    \begin{align}
        &\mathbb{E}\left[\sum_{e \in C^*} w_e \cdot x_e\right] \geq \sum_{v \in V} (1-\epsilon)(1-\epsilon)(1-\delta-q_v^N)\varphi_v^I \nonumber \\ &\geq (1-2\epsilon)\sum_{v \in V}(1-\delta-q_v^N)\varphi_v^I \nonumber \\
        &= (1-2\epsilon)(1-\delta)\varphi(C^*) - (1-2\epsilon)\sum_{v \in V} q_v^N \varphi_v^I \label{ineq:thm3.5.1}
    \end{align}
    
    where the last line follows from the definition of $\varphi(C^*)$. So, for the entire matching using Lemma \ref{lem:nbound} and Inequality \ref{ineq:thm3.5.1} we have
    \begin{align*}
        \mathbb{E}\left[\sum_{e \in E} x_e\right] &\geq \mathbb{E}\left[\sum_{e \in N} w_e \cdot x_e\right] + \mathbb{E}\left[\sum_{e \in C^*} w_e \cdot x_e\right] \\
        &\geq (1-10\epsilon)\varphi(N) \\ &+(1-2\epsilon)\left((1-\delta)\varphi(C^*) - \sum_{v \in V} q_v^N\varphi_v^I\right).
    \end{align*}
    Since $\varphi(N) \leq \mathbb{E}[\mu(\mathcal{G})]$,
    
    \begin{align*}
        &\mathbb{E}\left[\sum_{e \in E} x_e\right] - 8\epsilon\mathbb{E}[\mu(\mathcal{G})] \\&\geq (1-2\epsilon)\left(\varphi(N) + (1-\delta)\varphi(C^*) - \sum_{v \in V} q_v^N \varphi_v^I\right).
    \end{align*}
    
    As in the unweighted case, we divide both sides by the expected matching weight of non-crucial and crucial edges:
    \begin{align}
        &\frac{\mathbb{E}\left[\sum_{e \in E} x_e\right] - 8\epsilon\mathbb{E}[\mu(\mathcal{G})]}{\varphi(N) + \varphi(C^*)} \nonumber \\
        &\geq \frac{(1-2\epsilon)\left(\varphi(N) + (1-\delta)\varphi(C^*) - \sum_{v \in V} q_v^N \varphi_v^I\right)}{\varphi(N) + \varphi(C^*)} \nonumber  \\
        &\geq (1-2\epsilon)\left(1 - \frac{\sum_{v \in V} \delta \cdot \varphi_v^I + q_v^N\varphi_v^I}{\varphi(N) + \varphi(C^*)}\right) \nonumber \\
        &\geq (1-2\epsilon)\left(1 - \frac{\sum_{v \in V} \delta \cdot \varphi_v^I + q_v^N\varphi_v^I}{\sum_{v \in V} \varphi_v^I + \frac{\varphi_v^N}{2}}\right). \label{ineq:total}
    \end{align}
    
    From Lemma $\ref{lem:weightedw}$, for edge $e \in C^*$ directed to vertex $v$, $w_e \leq 2(1+\delta)\varphi_v^N$.  So, 
    \[
        \varphi_v^I \leq 2(1+\delta)q_v^C\varphi_v^N \leq 2(1+\delta)(1-q_v^N)\varphi_v^N.
    \]
    Moreover, the fraction 
    \[
        \frac{\delta \cdot \varphi_v^I + q_v^N\varphi_v^I}{\varphi_v^I + \frac{\varphi_v^N}{2}}
    \]
    is increasing with $\varphi_v^I$, and thus, by substituting in $2(1+\delta)(1-q_v^N)\varphi_v^N$ as an upper bound for $\varphi_v^N$, the fraction is at most
    
    \begin{align} & \frac{\delta \cdot \varphi_v^I + q_v^N\varphi_v^I}{\varphi_v^I + \frac{\varphi_v^N}{2}} \leq \nonumber \\
        &\frac{2(1+\delta)(1-q_v^N)\varphi_v^N(\delta+q_v^N)}{\varphi_v^N(\frac{1}{2} + 2(1+\delta)(1-q_v^N))} \nonumber \\
        &= \frac{2(1+\delta)(1-q_v^N)(\delta+q_v^N)}{\frac{1}{2} + 2(1+\delta)(1-q_v^N)}. \label{eq:deltas}
    \end{align}
    As defined, we can substitute $0.09$ for $\delta$ into equation $\ref{eq:deltas}$ to obtain
    \[
         \frac{-2.18 (q_v^N)^2 + 1.9338 q_v^N + 0.1962}{-2.18q_v^N + 2.68}.
    \]
    This fraction can be found to be at most $0.43$, however we skip the lengthy calculations for the sake of brevity.  We will next rely on the following algebraic lemma, whose proof is trivial and excluded.
    \begin{lemma}
\label{lem:algebraic2}
    For positive real values a, b, c, d, $\alpha$, if $\frac{a}{b} \leq \alpha$ and $\frac{c}{d} \leq \alpha$, then $\frac{a + c}{b+ d} \leq \alpha$.
\end{lemma}
Then, by Lemma $\ref{lem:algebraic2}$ and inequalities \ref{ineq:total} and \ref{eq:deltas}, we have that
    \begin{align}
        &\frac{\mathbb{E}\left[\sum_{e \in E} x_e\right] - 8\epsilon\mathbb{E}[\mu(\mathcal{G})]}{\varphi(N) + \varphi(C^*)} \nonumber \\
        &\geq (1-2\epsilon)\left(1 - \frac{\sum_v \delta \cdot \varphi_v^I + q_v^N\varphi_v^I}{\sum_v \varphi_v^I + \frac{\varphi_v^N}{2}}\right) \nonumber \\
        &\geq (1-2\epsilon)(1 - .43) = (1-2\epsilon)(.57). \label{ineq:thm3.5.2}
    \end{align}
    
    Finally, we can rearrange the terms of Inequality \ref{ineq:thm3.5.2} and using our assumption that $\varphi(C^*) > .926\varphi(C)$, prove the final bound.
    \begin{align*}
        &\mathbb{E}\left[\sum_{e \in E} x_e\right] \geq (1-2\epsilon)0.57(\varphi(N)+\varphi(C^*)) + 8\epsilon \mathbb{E}[\mu(\mathcal{G})] \\
        &\geq (1-2\epsilon)(0.57\cdot0.926)(\varphi(N)+\varphi(C)) - 8\epsilon\mathbb{E}[\mu(\mathcal{G})] \\
        &\geq (1-2\epsilon)(0.528)(\varphi(N)+\varphi(C)) - 8\epsilon\mathbb{E}[\mu(\mathcal{G})] \\
        &\geq (.528 - 2\epsilon)\mathbb{E}[\mu(\mathcal{G})] - 8\epsilon\mathbb{E}[\mu(\mathcal{G})] \\
        &= (.528 - 10\epsilon)\mathbb{E}[\mu(\mathcal{G})]
    \end{align*}
    As done previously, note that we can apply Lemma \ref{lem:frac2int} since we satisfy its requirements by Lemma \ref{lem:weightedblossom}.  Thus,
    \begin{align*}
        \mathbb{E}[\mu(\mathcal{Q})] \geq (1-\epsilon)\mathbb{E}\left[\sum_{e \in E} x_e\right] \geq (.528 - 11\epsilon)\mathbb{E}[\mu(\mathcal{G})].
    \end{align*}
\end{proof}

\section{Missing Proof of EDCS 2/3 Approximation}
\label{appx:edcs}

\begin{proof} [Proof of Lemma \ref{lem:edcs-main}]
We show the existence of two sub-graphs $\tilde{\mathcal{Q}} \subseteq \mathcal{Q}$ and $\tilde{\mathcal{G}} \subseteq \mathcal{G}$ with the following properties.
\begin{enumerate}
    \item \label{enum:edcs-p1} $\mathbb{E}[\mu(\tilde{\mathcal{G}})] \ge (1-\epsilon) \mathbb{E}[\mu(\mathcal{G})]$, where the expectation is taken over the realization of graph. 
    \item \label{enum-edcs-p2} $\tilde{\mathcal{Q}}$ is an EDCS$(\tilde{\mathcal{G}},(1+\epsilon) p_v \cdot p_e \cdot \beta, (1-2\epsilon) p_v \cdot p_e \cdot \beta)$ for $\tilde{\mathcal{G}}$.
\end{enumerate}
First we show how the existence of sub-graphs $\tilde{\mathcal{Q}}$ and $\tilde{\mathcal{G}}$ implies the lemma. By (\ref{enum-edcs-p2}), we have that  $\tilde{\mathcal{Q}}$ is an EDCS for $\tilde{\mathcal{G}}$. Also, $\frac{(1+\epsilon) p_v \cdot p_e \cdot \beta}{(1-2\epsilon) p_v \cdot p_e \cdot \beta} = 1+ O(\epsilon)$, and $(1+\epsilon) p_v \cdot p_e \cdot \beta = \Omega(\epsilon^2 \log(1/\epsilon))$. Therefore, by Theorem \ref{thm:edcs}, we have $\mu(\mathcal{Q}) \ge \mu(\tilde{\mathcal{Q}}) \ge (2/3 - O(\epsilon)) \mu(\tilde{\mathcal{G}})$. Combining this with the property (\ref{enum:edcs-p1}), gives us $\mathbb{E} [\mu(\mathcal{Q})] \ge (2/3 - O(\epsilon)) \mathbb{E}[\mu(\mathcal{G})]$ which concludes the lemma.

Consider a vertex $v \in V$, this vertex is in the realized sub-graph $\mathcal{G}$ with the probability of $p_v$. Therefore with the probability of $1-p_v$, vertex $v$ is not realized and we have $deg_\mathcal{Q}(v) =0$. Consider the case that $v$ is realized. We then have $\mathbb{E}[deg_\mathcal{Q}(v)] = p_e \cdot p_v \cdot deg_{Q}(v)$, since neighbors of $v$ are realized with the probability of $p_v$ and incident edges of $v$ are realized with the probability of $p_e$. In the following lemma we show that only a small fraction of realized vertices can significantly deviate from this expectation.

\begin{definition}
Let $\mathcal{V}^+ \subseteq \mathcal{V}$ be the set of realized vertices $v$ such that  $deg_{\mathcal{Q}}(v) > p_e \cdot p_v \cdot deg_{Q}(v) + \epsilon \cdot p_v \cdot p_e \cdot \beta/2$. Also, let $\mathcal{V}^- \subseteq \mathcal{V}$ be the set of realized vertices $v$ such that  $deg_{\mathcal{Q}}(v) < p_e \cdot p_v \cdot deg_{Q}(v) - \epsilon \cdot p_v \cdot p_e \cdot \beta/2$ or there exists an edge $(v,u) \in Q$ such that $u \in \mathcal{V}^+$. 
\end{definition}

\begin{lemma}
\label{lem:vpvm}
$\mathbb{E}[|\mathcal{V}^+|], \mathbb{E}[|\mathcal{V}^-|] \le \epsilon^6 \cdot p_v^6 \cdot p_e^6 \cdot \mu(G)$.
\end{lemma}
\begin{proof}
Consider a realized vertex $v \in \mathcal{V}$. We know that $deg_Q(v) \le \beta$. Consider an edge $(v,u) \in Q$. This edge appears in the realized sub-graph with the probability of $p_v \cdot p_e$ since vertex $u$ is realized with the probability of $p_v$, and edge $(v,u)$ is realized with the probability of $p_e$. Therefore we have $\mathbb{E}[deg_{\mathcal{Q}}(v)] = p_v \cdot p_e \cdot deg_Q(v) \le p_v \cdot p_e \cdot \beta$. By Chernoff bound we have
\begin{align*}
    &\Pr [ v \in \mathcal{V}^+] = \Pr[ v \text{ is realized}] \\
    &\cdot \Pr[deg_{\mathcal{Q}}(v) > p_e \cdot p_v \cdot deg_{Q}(v) + \epsilon \cdot p_v \cdot p_e \cdot \beta/2]\\
    &\le p_v\cdot e^{-O(\epsilon^2 \cdot p_v \cdot p_e \cdot \beta)} \\
    &\le e^{-O(\log (1/(\epsilon \cdot  p_v \cdot p_e)))} \le K^{-2} \cdot \epsilon^{12} \cdot p_v^{12} \cdot p_e^{12} \,,
\end{align*}
where $K$ is a large constant and the last two inequalities follow from the fact that $\beta = \frac{C \log(1/(\epsilon \cdot  p_v \cdot p_e))}{{\epsilon^2 p_v p_e}}$ for a large $C$. We can also set $C$ large enough to get an arbitrary large constant $K >C $. Therefore, each vertex is in $\mathcal{V}^+$ with the probability of  $K^{-2} \cdot \epsilon^{12} \cdot p_v^{12} \cdot p_e^{12}$. On the other hand, $Q$ has at most $2 \beta \mu(G)$ vertices with non-zero degree. The reason is that the graph $G$ has a vertex cover of size $2 \mu(G)$ and all vertices in $Q$ have a degree of at most $\beta$. Since $Q$ has at most $2 \beta \mu(G)$ vertices with non-zero degree, and each of these vertices are in is in $\mathcal{V}^+$ with the probability of  $K^{-2} \cdot \epsilon^{12} \cdot p_v^{12} \cdot p_e^{12}$, we have
\begin{align*}
\mathbb{E}[|\mathcal{V}^+|] &\le   2 \beta \mu(G) \cdot K^{-2} \cdot \epsilon^{12} \cdot p_v^{12} \cdot p_e^{12} \\
&\le \mu(G) \cdot K^{-2} \cdot C \cdot \epsilon^{9} \cdot p_v^{9}\cdot  p_e^{9}\\
&\le \mu(G) \cdot K^{-1} \cdot  \epsilon^{9} \cdot p_v^{9} \cdot p_e^{9} \,, &\text{Since $K>C$.}
\end{align*}
which shows that the number of vertices in $|\mathcal{V}^+|$ is small.
Using a similar argument we can say that the number of vertices with the degree less than $p_e \cdot p_v \cdot deg_{Q}(v) - \epsilon \cdot p_v \cdot p_e \cdot \beta/2$ is at most $ K^{-1} \cdot  \epsilon^{9} \cdot p_v^{9} \cdot p_e^{9}$. Since $\mathcal{V}^-$ is the set of low-degree vertices and neighbors of $\mathcal{V}^+$ in $Q$, we have
\begin{align*}
    \mathbb{E}[|\mathcal{V}^-|] &\le \mu(G) \cdot K^{-1} \cdot  \epsilon^{9} \cdot p_v^{9} \cdot p_e^{9} + \beta \cdot \mathbb{E}[|\mathcal{V}^+|] \\
    &\le \mu(G) \cdot  \epsilon^{6} \cdot p_v^{6} \cdot p_e^{6} \,,
\end{align*}
which proves the lemma.
\end{proof}
Lemma \ref{lem:vpvm} above shows that the sizes of $\mathcal{V}^+$ and $\mathcal{V}^-$ are very small. We complete the proof of Lemma \ref{lem:edcs-main} by constructing sub-graphs $\tilde{\mathcal{Q}}$ and $\tilde{\mathcal{G}}$ as follows. Let $\tilde{\mathcal{G}}$ have vertex set $\mathcal{V}$ which are the set of realized vertices, and have the edge set equal to $\mathcal{G}$, except we remove all edges incident to vertices in $\mathcal{V}^+$, and all edges $(v,u) \notin Q$ that are incident to vertices in   $\mathcal{V}^-$. Also let $\tilde{\mathcal{Q}}$ be the same as $\mathcal{Q}$ except we remove all edges incident to $\mathcal{V}^+$. Now we show that these sub-graphs satisfy properties (\ref{enum:edcs-p1}) and (\ref{enum-edcs-p2}).

For property (\ref{enum:edcs-p1}), note that $\tilde{\mathcal{G}}$ and $\mathcal{G}$ are different only in the vertices in $\mathcal{V}^+$ and $\mathcal{V}^-$. Therefore,
\begin{align*}
    \mathbb{E}[\mu(\tilde{\mathcal{G}})] &\ge \mathbb{E}[\mu(\mathcal{G})] - \mathbb{E}[|\mathcal{V}^+|] -\mathbb{E}[|\mathcal{V}^-|] \\
    &\ge \mathbb{E}[\mu(\mathcal{G})] - \epsilon^5 \cdot p_v^5 \cdot p_e^5 \cdot \mu(G)  & \text{By Lemma \ref{lem:vpvm}.} \\
    &\ge \mathbb{E}[\mu(\mathcal{G})] - \epsilon^3 \cdot p_v^3 \cdot p_e^3 \cdot \mathbb{E}[\mu(\mathcal{G})] \\
    &\ge (1-\epsilon) \mathbb{E}[\mu(\mathcal{G})] ,,
\end{align*}
where the third inequality follows from the fact that $\mathbb{E}[\mu(\mathcal{G})] \ge p_v^2 \cdot p_e \cdot \mu(G)$, since every edge in a maximum matching is realized with the probability of $p_v^2 \cdot p_e$.

For property (\ref{enum-edcs-p2}), we have to show that $\tilde{\mathcal{Q}}$ is an EDCS$(\tilde{\mathcal{G}},(1+\epsilon) p_v \cdot p_e \cdot \beta, (1-2\epsilon) p_v \cdot p_e \cdot \beta)$ for $\tilde{\mathcal{G}}$. To that purpose, we show that $\tilde{\mathcal{Q}}$ satisfies properties (\ref{enum:mainedcs1}) and (\ref{enum:mainedcs2}) of Definition \ref{def:edcs}.
Both $\tilde{\mathcal{G}}$ and $\tilde{\mathcal{Q}}$ do not have any edges incident to $\mathcal{V}^+$. Therefore, we can ignore these vertices. Therefore, for all vertices $v$ we have $deg_{\tilde{\mathcal{Q}}}(v) \le p_e \cdot p_v \cdot deg_{Q}(v) + \epsilon \cdot p_v \cdot p_e \cdot \beta/2$, and for the vertices $v \notin \mathcal{V}^-$ we have $deg_{\tilde{\mathcal{Q}}}(v) \ge p_e \cdot p_v \cdot deg_{Q}(v) - \epsilon \cdot p_v \cdot p_e \cdot \beta/2$. Also $\tilde{\mathcal{G}} \setminus \tilde{\mathcal{Q}}$ has no edge incident to $\mathcal{V}^-$.
\begin{enumerate}
    \item For the property (\ref{enum:mainedcs1}) of Definition \ref{def:edcs}: Consider an edge $(v,u) \in \tilde{\mathcal{Q}}$, we then have
    \begin{align*}
        &deg_{\tilde{\mathcal{Q}}}(v) +  deg_{\tilde{\mathcal{Q}}}(u) \\
        &\le p_e \cdot p_v \cdot deg_{Q}(v) \\
        &+ p_e \cdot p_v \cdot deg_{Q}(u) \\
        &+ \epsilon \cdot p_v \cdot p_e \cdot \beta \\
        &\le (1+\epsilon) p_v \cdot p_e \cdot \beta \,. & \text{Since $Q$ is an EDCS of $G$.}
    \end{align*}
        \item For the property (\ref{enum:mainedcs2}) of Definition \ref{def:edcs}: Consider an edge $(v,u) \in \tilde{\mathcal{G}} \setminus \tilde{\mathcal{Q}}$, we then have
    \begin{align*}
        &deg_{\tilde{\mathcal{Q}}}(v) +  deg_{\tilde{\mathcal{Q}}}(u) \\
        &\ge p_e \cdot p_v \cdot deg_{Q}(v) \\
        &+ p_e \cdot p_v \cdot deg_{Q}(u) \\
        &- \epsilon \cdot p_v \cdot p_e \cdot \beta \\
        &\ge (1-2\epsilon) p_v \cdot p_e \cdot \beta \,, & \text{Since $Q$ is an EDCS of $G$.}
    \end{align*}
\end{enumerate}
which completes the proof of Lemma \ref{lem:edcs-main}.
\end{proof}

\end{document}